\documentclass[11pt,a4paper]{article}
\pdfoutput=1
\usepackage{microtype}

\usepackage[colorlinks=true,urlcolor=Blue,citecolor=Green,linkcolor=BrickRed]{hyperref}
\usepackage[usenames,dvipsnames]{xcolor}

\usepackage[backend=biber,maxbibnames=99]{biblatex}
\renewbibmacro{in:}{\ifentrytype{article}{}{\printtext{\bibstring{in} }}} %

\bibliography{biblio}

\newbibmacro{title+doi+url}[1]{
  \iffieldundef{doi}{
    \iffieldundef{url}{#1}
    {\href{\thefield{url}}{#1}}
  } {\href{http://dx.doi.org/\thefield{doi}}{#1}}
}
\DeclareFieldFormat*{title}{\usebibmacro{title+doi+url}{\mkbibemph{#1}}}

\usepackage[T1]{fontenc}
\usepackage{amsthm}
\usepackage{amssymb}
\usepackage{amsmath}
\usepackage{amsfonts}
\usepackage{todonotes}
\usepackage{authblk}
\usepackage{xspace}
\usepackage{fullpage}
\usepackage{enumerate}%
\usepackage{relsize}

\newcommand{\Oh}{\mathcal{O}}
\newcommand{\Ohtilda}{\tilde{\Oh}}
\newcommand{\eps}{\varepsilon}
\newcommand{\mult}{\textsc{mult}}
\newcommand{\transform}[3]{#1+(#2,#3)}
\newcommand{\MBR}{\textsc{MBR}}

\newcommand{\rev}{\textsc{rev}}
\newcommand{\TT}{\mathcal{T}}
\newcommand{\NN}{\mathcal{N}}
\newcommand{\shape}[4]{\scriptstyle{\frac{#1|#2}{#3|#4}}}
\newcommand{\bigshape}[4]{{\frac{#1|#2}{#3|#4}}}

\newtheorem{theorem}{Theorem}[section]

\newtheorem{corollary}[theorem]{Corollary}
\newtheorem{lemma}[theorem]{Lemma}

\theoremstyle{definition}   
\newtheorem{definition}[theorem]{Definition}

\newtheorem{observation}[theorem]{Observation}
\newtheorem{conjecture}{Conjecture}
\newtheorem{question}{Question}

\renewcommand{\paragraph}{\subparagraph}

\title{Counting 4-Patterns in Permutations Is Equivalent to Counting 4-Cycles in Graphs}
\date{}	
\author[1]{Bart{\l}omiej Dudek}
\author[1]{Pawe{\l} Gawrychowski}

\affil[1]{Institute of Computer Science, University of Wroc{\l}aw, Poland}

\newcommand{\FIGURE}[4]{
\begin{figure}[#1]
\begin{centering}
 \includegraphics[width={#2}\textwidth]{{figures/#3}.pdf}
\caption{#4}
\label{fig:#3}
\end{centering}
\end{figure}
}

\begin{document}
\maketitle
  
\begin{abstract}
Permutation $\sigma$ appears in permutation $\pi$ if there exists a subsequence of $\pi$ that is order-isomorphic to $\sigma$.
The natural algorithmic question is to check if $\sigma$ appears in $\pi$, and if so count the number of occurrences.
Only since very recently we know that for any fixed length~$k$, we can check if a given pattern of length $k$ appears in a permutation of length $n$ in time linear in $n$, but being able to count all such occurrences in $f(k)\cdot n^{o(k/\log k)}$ time would refute the exponential time hypothesis (ETH).
Together with practical applications in statistics, this motivates a systematic study of the complexity of counting occurrences
for different patterns of fixed small length $k$.
We investigate this question for $k=4$.
Very recently, Even-Zohar and Leng [arXiv 2019] identified two types of $4$-patterns.
For the first type they designed an $\Ohtilda(n)$ time algorithm\footnote{$\Ohtilda(.)$ hides factors polylogarithmic in $n$.}, while for the second they were able to provide an $\Ohtilda(n^{1.5})$ time algorithm.
This brings up the question whether the permutations of the second type are inherently harder than the first type.

We establish a connection between counting 4-patterns of the second type and counting 4-cycles (not necessarily induced) in a sparse undirected graph.
By designing two-way reductions we show that the complexities of both problems are the same, up to polylogarithmic factors.
This allows us to leverage the work done on the latter to provide a reasonable argument for why there is a difference
in the complexities for counting 4-patterns of the first and the second type.
In particular, even for the seemingly simpler problem of detecting a 4-cycle in a graph on $m$ edges, the best known algorithm works
in $\Oh(m^{4/3})$ time. Our reductions imply that an $\Oh(n^{4/3-\eps})$ time algorithm for counting occurrences of any 4-pattern
of the second type in a permutation of length $n$ would imply an exciting breakthrough for counting (and hence also detecting) 4-cycles. In the other direction, by plugging
in the fastest known algorithm for counting 4-cycles, we obtain an algorithm for counting occurrences of any 4-pattern of the second type in $\Oh(n^{1.48})$ time.
\end{abstract}
 
\section{Introduction}\label{se:introduction}

Permutations are arguably the most basic combinatorial objects. A natural question in discrete mathematics is to
count permutations with certain properties, like consisting of a given number of cycles or having no fixed points.
A whole class of such questions is obtained by fixing a permutation $\sigma$, called the pattern, and defining
a permutation $\pi$ to avoid $\sigma$ if $\sigma$ is not a sub-permutation of $\pi$, or in other words
if $\pi$ does not contain a subsequence that is order-isomorphic to $\sigma$. For example, $21$
is avoided only by $12\ldots n$. Otherwise, we say that $\pi$ contains $\sigma$.
One of the first results concerning pattern avoidance is by Erd\H{o}s and Szekeres \cite{ErdosS35}, who proved
that every permutation of at least $(k-1)(\ell-1)+1$ elements contains either $12\cdots k$ or $\ell \cdots 21$.
Another classical result in pattern avoidance is due to Knuth~\cite{Knuth68}, who showed
that $\pi$ can be sorted by a stack if and only if $\pi$ avoids $231$.
Together with the systematic study of patterns in permutations by Simion and Schmidt~\cite{SimionS85},
this sparked an interest in counting and characterising permutations that avoid a given pattern (or multiple patterns).
A remarkable result in this area is by Marcus and Tardos~\cite{MarcusT04}, who showed that the number
of permutations of length $n$ avoiding $\sigma$ is bounded by $c(\sigma)^n$, where $c(\sigma)$ is a function independent of $n$.
This was conjectured in early 1990s independently by Stanley and Wilf.
For further discussion we refer the reader to surveys and textbooks~\cite{Vatter15,Bona12,Kitaev11}.

We approach pattern avoidance from an algorithmic perspective. We cannot hope for an efficient algorithm for arbitrary
patterns, as in general it is NP-hard to check if $\pi$ contains~$\sigma$~\cite{BoseBL98} when $\sigma$ is part of the input.
However, if we restrict our attention to patterns of length $k$, we might hope to check if a given permutation
on $n$ elements avoids such pattern faster than using the trivial algorithm in $\Oh(n^{k})$ time.
Indeed, Albert et al.~\cite{AlbertAAH01} and Ahal and Rabinovich~\cite{AhalR08} improved this complexity
to $\Oh(n^{2k/3+1})$ and $n^{0.47k+o(k)}$, respectively.
In a recent breakthrough result, Guillemot and Marx~\cite{GuillemotM14} developed a fixed-parameter tractable (FPT) algorithm that runs
in $2^{\Oh(k^2\log k)}\cdot n$ time.
Later, by refining the proof of Marcus and Tardos~\cite{MarcusT04}, Fox~\cite{Fox13} removed the $\log k$ factor
in the exponent to arrive at $2^{\Oh(k^2)}\cdot n$ complexity.
For $k\geq n/\log n$, $\Oh(1.79^n)$ and $\Oh(1.618^n)$ time algorithms are known~\cite{BrunerL16,BerendsohnKM19}.
Hence even though the problem is NP-hard, by now we have a range of efficient algorithms for different
special cases of checking pattern avoidance.

However, some applications bring the need to not only detect but also count occurrences of the pattern.
A basic example is calculating the so-called Kendall's $\tau$ correlation coefficient~\cite{Kendall38}, which
requires counting inversions. Generalizations of Kendall's test used in statistics require counting occurrences
of larger patterns.
Bergsma-Dassios~\cite{BergsmaD14} and Yanagimoto~\cite{Yanagimoto70}
used patterns of length 4 in their tests. Finally, patterns of length 5 appear in the Hoeffding's dependence
coefficient~\cite{Hoeffding48}. Also see Heller et al. \cite{HellerHKBG16} for a general family of such tests.
We refer the reader to \cite{Even-ZoharL19} for a more detailed description of the viewpoint of permutations in nonparametric statistics of bivariate data.
Unfortunately, hardly any of the aforementioned algorithms for detecting patterns generalize to counting.
A recent result by Berendsohn et al. \cite{BerendsohnKM19} shows that this is, in fact, inevitable,
as if patterns of length $k$ can be counted in $f(k)n^{o(k/\log k)}$ time then the exponential-time hypothesis fails.
This shows that we cannot hope for a general FPT algorithm, and considering the applications in statistics we should
focus on understanding the best possible exponent for small values of $k$.

Patterns of length $k$ can be trivially counted in $\Oh(n^k)$ time, which was improved by Albert et al.~\cite{AlbertAAH01}
to $\Oh(n^{2k/3+1})$ and then by Berendsohn et al.~\cite{BerendsohnKM19} to $\Oh(n^{k/4+o(k)})$ time.
However, it is clear that among all patterns of the same length $k$ some are easier to count than the others.
For example, occurrences of $12\cdots k$ can be easily counted in $\Ohtilda(nk)$ time using dynamic programming and range queries.
This motivates a systematic study of the complexity of counting occurrences of different patterns of fixed small length. 
For $k=2$, this is exactly the well-known exercise of counting inversions (or in other words, the pattern $21$) in a permutation (or its reverse), which
can be solved in $\Oh(n\log n)$ time with merge sort or in $\Oh(n\sqrt{\log n})$ in the Word RAM model~\cite{ChanP10}.
For $k=3$, all patterns can be counted in $\Ohtilda(n)$ time by using appropriate range counting structures.
For $k=4$, various algorithms were designed to compute efficiently the Bergsma-Dassios test, which asks about the value  $\tau^*=({\scriptstyle \#1234+ \#1243+ \#2134+ \#2143+ \#3412+ \#3421+ \#4312+ \#4321 })/ {n \choose 4} - \frac13$ \cite{BergsmaD14}.
First approaches brought the complexity down to $\Oh(n^2)$ \cite{WeihsDM18,WeihsDL16,HellerH16} and finally, very recently, Even-Zohar and Leng~\cite{Even-ZoharL19} observed that the patterns counted in this test possess some structural property that allows to design an $\Ohtilda(n)$ time algorithm.
For the remaining patterns of size 4, they obtained an algorithm working in $\Ohtilda(n^{1.5})$ time. 
Defining the $k$-profile of a permutation $\pi$ to be the sequence of $k!$ numbers with the number of occurrences for every possible pattern $\sigma$ of length $k$,
this brings us to the following natural open question:

\begin{question}[Even-Zohar and Leng \cite{Even-ZoharL19}]
\label{q:complexity}
What is the computational complexity of finding the full 4-profile of a given permutation of length $n$?
\end{question}

In fact, Even-Zohar and Leng~\cite{Even-ZoharL19} showed that among all the twenty-four 4-patterns, there are eight that can be counted in $\Ohtilda(n)$ time, while the remaining ones can be counted in $\Ohtilda(n^{1.5})$ time.
Additionally, they showed that all patterns of the second type are equivalent in terms of computational complexity, that is after counting one of them, we can retrieve all the other in $\Ohtilda(n)$ time.
These two types in fact coincide with the notion of concordant and discordant patterns as defined by Bergsma and Dassios~\cite{BergsmaD14}.
Using the notation of Fox \cite{Fox13}, the permutation matrix of patterns of the second type contains $J_2$ as an interval minor.
This raises the challenge of finding a reason why some 4-patterns seem harder to count than the others.

\begin{question}\label{q:harder}
Why some 4-patterns seem more difficult to count than the others?
\end{question}

\paragraph{Related work.}
Many efforts have been devoted to understand which patterns are more difficult to
detect~\cite{AlbertLLV16,BoseBL98,AlbertAAH01,Ibarra97,YugandharS05,GuillemotV09}.
Recently Jel{\'{\i}}nek and Kyn\v{c}l \cite{JelinekK17} established that it is possible to detect $\sigma$ in polynomial time
if $\sigma$ avoids $\alpha$, for $\alpha\in\{1,12,21,132,213,231,312\}$ and NP-complete otherwise. This was later
strengthened by Berendsohn et al.~\cite{BerendsohnKM19} by considering treewidth of the incidence graph of $\sigma$.
Even though the problem is NP-hard in general, more efficient algorithms are known for many families of patterns, such as
vincular~\cite{BabsonS00}, bivincular~\cite{Bousquet-MelouCDK10}, mesh~\cite{BrandenC11}, boxed mesh~\cite{AvgustinovichKV13}
and consecutive~\cite{ElizaldeN03}. See the survey by Bruner and Lackner~\cite{BrunerL13} for a more detailed description of these variants.

\paragraph{Fine-grained complexity.}
Although the traditional notion of ``easy'' and ``hard'' problems is defined with respect to the polynomial time solvability, in the last two decades commenced the study on ``fine-grained'' theory which tries to understand relationships between polynomial-time solvable problems.
They can be employed to state conditional lower bounds based on one of a few believable conjectures on complexities of some basic problems,
such as SETH, APSP, or 3SUM.
See a recent survey by Vassilevska Williams~\cite{Williams14} for a summary.

\paragraph{Counting short cycles in graphs.}
Similarly as for permutation patterns, a natural question is to detect or count small substructures of a graph,
with perhaps the most fundamental example being counting cycles of particular length.
Already the smallest case, triangle, is highly non-trivial to count, as the fastest known approach for a $n$-node
graph runs in $\Oh(n^{\omega})=\Oh(n^{2.38})$ using fast matrix multiplication algorithm~\cite{Gall14a,Williams12}.

Surprisingly, Vassilevska Williams and Williams~\cite{WilliamsW18} proved that this is essentially inevitable, as the
two problems are, in a certain sense, equivalent: a practical advance for detecting triangles would imply a practical
algorithm for Boolean matrix multiplication.
As in many applications the graphs are sparse, it is desirable to design algorithm with running time depending on
the number of edges $m$. Alon et al.~\cite{AlonYZ97} developed an $\Oh(m^{2\omega/(\omega+1)})=\Oh(m^{1.41})$
time algorithm for counting triangles (in fact their algorithm is stated for finding a single triangle, but can be easily extended).
Going one step further, 4-cycles can also be counted in $\Oh(n^{\omega})$ time~\cite{AlonYZ97}, but the fastest known
counting algorithm for sparse graphs runs in $\Oh(m^{1.48})$ time~\cite{WilliamsWWY15}.
Interestingly, one can {\it find} a $2k$-cycle, for any constant $k\geq 2$, in $\Oh(n^{2})$ time~\cite{YusterZ97}.
If the graph is given as an adjacency matrix, this is clearly optimal,
but it seems plausible to conjecture that this is also the case if the graph is given as adjacency lists.

\begin{conjecture}[Yuster and Zwick~\cite{YusterZ97}]\label{conj:no_n2-eps}
 For every $\eps>0$, there is no algorithm that detects 4-cycles in a graph on $n$ nodes in~$\Oh(n^{2-\eps})$ time.
\end{conjecture}

The best known algorithm for finding a 4-cycle in a sparse graph runs in $\Oh(m^{4/3})$ time~\cite{AlonYZ97}.
This was recently extended by Dahlgaard et al.~\cite{DahlgaardKS17} who showed how to find a $2k$-cycle in
$\Oh(m^{2k/(k+1)})$ time.
Furthermore, they showed that this is in fact optimal, assuming Conjecture~\ref{conj:no_n2-eps} and using a general combinatorial result of Bondy and Simonovits that a graph with $m=100kn^{1+1/k}$ edges must contain a $2k$-cycle \cite{BondyS74}. 
See also Abboud and Vassilevska Williams~\cite{AbboudW14} for a similar conjecture on the complexity of detecting a 3-cycle.

\begin{conjecture}[Dahlgaard, Knudsen and St{\"{o}}ckel~\cite{DahlgaardKS17}]\label{conj:no_m4/3-eps}
For every $\eps>0$, there is no algorithm that detects a 4-cycle in a graph with $m$ edges in~$\Oh(m^{4/3-\eps})$ time.
\end{conjecture}

Dudek and Gawrychowski \cite{DudekG19} recently used this conjecture to provide an explanation for why there is
no $\Ohtilda(n)$ time algorithm for computing the so-called quartet distance between two trees on $n$ nodes.
Very recently Duraj et al.  developed an equivalence class between range query problems and detecting triangles in sparse graphs~\cite{DurajKPW20}.

\paragraph{Our contribution.}
As in the previous works we divide the patterns into two types and we call them trivial and non-trivial respectively.
Our main contribution is a two-way reduction between counting occurrences of a non-trivial pattern and counting 4-cycles in an undirected sparse graph.
This provides a reasonable answer for Question~\ref{q:harder}, as any $\Ohtilda(n)$ time algorithm for such patterns would imply an exciting
breakthrough for counting 4-cycles,
and confirms that the two types of 4-patterns identified in the previous work are inherently different.

We partially answer Question~\ref{q:complexity} about the exact complexity of computing 4-profile of permutation of length $n$.
Our two-way reductions imply that, by plugging in the asymptotically faster known algorithm for counting 4-cycles
in a sparse graph~\cite{WilliamsWWY15}, we are able to compute the full 4-profile of a permutation of length $n$
in $\Oh(n^{1.48})$ time. In the other direction, we argue that an $\Oh(n^{4/3-\eps})$ time algorithm is unlikely,
as long as one is willing to believe Conjecture~\ref{conj:no_m4/3-eps}.

Our reductions are summarised in Figure~\ref{fig:reductions}.
A corollary from these reductions is an alternative proof for the equivalence between the non-trivial patterns, which avoids using the notion of corner tree formulas and a computer-aided argument used in~\cite{Even-ZoharL19}.

\begin{figure}[b]

\newcommand{\myproblem}[1]{\large\textbf{#1}}
\newcommand{\mydesc}[1]{\small{Section {#1}}}
\newcommand{\myedgeshort}{edge[bend left=10]}
\newcommand{\myedgelong}{edge[bend left=15]}
\newcommand{\mynode}[2]{\node (#1) at #2 {#1}}
\newcommand{\mynodedisp}[3]{\node [align=left] (#1) at #2 {#3}}

\newcommand{\lem}[1]{\tiny{Lemma #1}}

    \centering
    \begin{tikzpicture}
    \begin{scope}[every node/.style={},
                  every edge/.style={draw=black,very thick}]
        
        \mynode{pattern}{(0,0)};
        \mynodedisp{4-partite pattern}{(3,0)}{4-partite\\pattern};
        \mynodedisp{4-circle-layered}{(7.5,0)}{4-circle-layered\\graph};
        \mynodedisp{undirected}{(12,0)}{undirected\\graph};
        \mynodedisp{4-circle-layered multigraph}{(5.25,-2)}{4-circle-layered\\multigraph};
        \mynodedisp{directed}{(9.75,-2)}{directed\\graph};
        
        \node (only lemma) at (9.75,-0.6) {\lem{\ref{le:undirected-and-directed-4-partite}}};
        
        \path [->] (pattern)   edge[bend left=20] node[above]  	{\lem{\ref{le:pattern-2-4-partite-pattern}}} (4-partite pattern);
        
        \path [->] (4-partite pattern)   edge[bend left=20] node[below]  	{\lem{\ref{le:4-partite-pattern-2-pattern}}} (pattern);
        
        \path [->] (4-partite pattern)   edge[bend right=20] node[left, near end]  	{\lem{\ref{le:pattern-2-cycle}}} (4-circle-layered multigraph);
        
        \path [->] (4-circle-layered multigraph)   edge[bend right=20] node[right, near start]  	{\lem{ \ref{le:dir_multi_to_simple}}} (4-circle-layered);
        
        \path [->] (4-circle-layered)   edge[bend right=20] node[above]  	{\lem{\ref{le:cycle-2-pattern}}} (4-partite pattern);
        
        \path [->] (4-circle-layered)   edge[bend left=20] node[right]  	{} (undirected);
        
        \path [->] (undirected)   edge[bend left=20] node[right]  	{} (directed);
        
        \path [->] (directed)   edge[bend left=20] node[right]  	{} (4-circle-layered);
        
    \end{scope}
    \end{tikzpicture}
    \caption{Sequence of reductions used to prove the equivalence between counting non-trivial 4-patterns and 4-cycles.
The right part of the figure describes different kinds of graphs in which we count 4-cycles.
}
    \label{fig:reductions}
\end{figure}
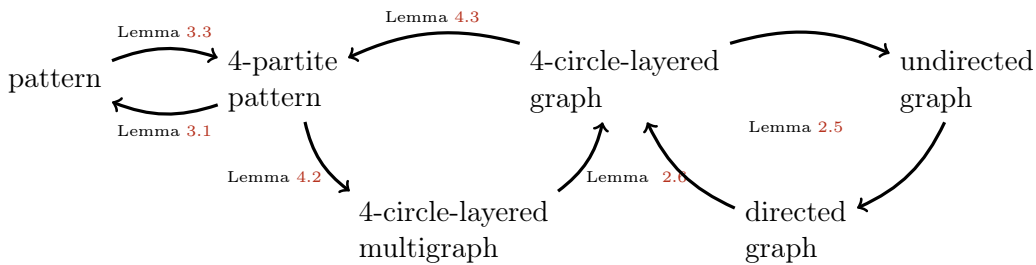

\begin{theorem}
 An algorithm for counting 4-cycles in a graph on $m$ edges in $\Ohtilda(m^\gamma)$ time implies an algorithm for counting non-trivial patterns in a permutation of length $n$ in $\Ohtilda(n^\gamma)$ time and vice versa.
\end{theorem}

We can plug in the fastest known algorithm for counting 4-cycles that runs in $\Oh(m^{\frac{4\omega-1}{2\omega+1}})=\Oh(m^{2-\frac{3}{2\omega+1}})$ time \cite{WilliamsWWY15}.
As $\omega<2.373$ \cite{Gall14a,Williams12}, we obtain a more efficient algorithm for computing the full 4-profile 
in  $\Oh(n^{1.48})$ time.

\begin{corollary}
 For every $\eps>0$, there exists no algorithm that can count non-trivial 4-patterns in permutation of length $n$ in $\Oh(n^{4/3-\eps})$ time unless
 Conjecture~\ref{conj:no_m4/3-eps} is false.
\end{corollary}

We stress that even though we use Conjecture~\ref{conj:no_m4/3-eps} about detecting 4-cycles, the reduction proceeds by
creating multiple instances and subtracting some of the obtained result. Hence, it does not imply anything about the complexity
of detecting 4-patterns, and in fact for this problem Guillemot and Marx \cite{GuillemotM14} showed an $\Oh(n)$ time algorithm.

\paragraph{Overview of the methods.}
Most of our reductions exploit the additional structure of pattern occurrences in the plane which is divided by a horizontal and a vertical line.
We group the occurrences by shapes corresponding to the number of points in each quadrant and count them separately.
It turns out that the hard case is when the four points are all in distinct quadrants. This is the heart of our main reductions
between counting patterns and 4-cycles.
All other shapes can be counted in almost linear time with a careful application of range queries.
To simplify the presentation, we split the reductions into many steps, between different classes of graphs and patterns so as to
work with 4-partite patterns and graphs which have more structure for our application.
Our reductions are based on the divide and conquer paradigm, applied to each of the four half-planes separately.
We present them using Minimum Base Ranges corresponding to nodes of the full binary tree on $n$ leaves.

Our reduction from counting 4-cycles to counting 4-patterns uses somewhat similar techniques to Berendsohn et al. \cite{BerendsohnKM19}.
However, their approach works for arbitrary subgraphs on $k$ nodes, which comes at a cost of increasing the size of permutation pattern and in our case would result in a pattern of 29 elements.
This would not give us the desired connection between counting 4-cycles and 4-patterns, so we need a new argument tailored for 4-cycles.

\section{Preliminaries}

Permutation $\pi$ of length $n$ is a bijective mapping $\pi: [n]\rightarrow [n]$, where $[n]=\{1,\ldots, n\}$ and a $k$-pattern $\sigma$ is a permutation of length $k$.
A permutation $\pi$ contains a $k$-pattern $\sigma$ if there exist indices $1\leq i_1<i_2<\ldots<i_k\leq n$ such that $\sigma(j)<\sigma(j')$ iff $\pi(i_j)<\pi(i_{j'})$ for distinct $j,j'\in[k]$.
A sequence of $k$ increasing indices with the above properties is called an \textit{occurrence} of $\sigma$ in $\pi$.
For example, in permutation $524\underline{6}\underline{1}7\underline{3}$ the underlined positions $4,5$ and $7$ form an occurrence of pattern $312$.
By counting a $k$-pattern in a permutation we mean counting occurrences of the pattern.
Unless stated otherwise, a pattern refers to a 4-pattern.

\paragraph{Shapes.}
We represent permutation $\pi$ as a set of points in the plane: $S_\pi=\{(i,\pi(i)): i\in [n]\}$ and we interchangeably use points and their corresponding elements from the permutation.
For instance, four points $\{(i_j,\pi(i_j)):j\in [4]\}$ are an occurrence of pattern $\sigma$ iff positions $i_1<\ldots<i_4$ are an occurrence of $\sigma$ in $\pi$.
We say that a horizontal line divides a plane into top and bottom part and vertical line divides into left and right part.
\textit{Division} of a plane with both horizontal and vertical line splits the points from $S_\pi$ into four \textit{regions} and we abbreviately denote each of them by capital letters denoting horizontal and vertical location of the region: TL,TR,BL or BR.
Slightly abusing the notation, by a region  we mean either the region or the set of points from $S_\pi$ that belong to the region, with the appropriate order between them.
Returning to the correspondence between the elements of $\pi$ and $S_\pi$, notice that the division of the plane with horizontal line $y=h$ and vertical line $x=v$ also partitions elements from $\pi$ into four groups, for instance $(i,\pi(i))\in$ TL iff $i<v \wedge \pi(i)>h$.
We will only consider such divisions of the plane that the dividing lines never pass through a point from $S_\pi$.

Given a division of the plane, we say that an occurrence of pattern $\sigma$ forms \textit{shape} $\shape abcd$ if among the 4 points, there are respectively $a,b,c$ and $d$ points in top-left, top-right, bottom-left and bottom-right region of the plane. 
By counting a particular shape for a division we mean counting the number of quadruples of points forming the shape with appropriate number of points in each of the regions.
Note that one pattern may form multiple shapes, i.e. $\shape1102,\shape1201,\shape2020$ or $\shape1111$, depending on the pattern and the position of the dividing lines.
However, some shapes cannot be formed by all patterns, no matter how we divide the plane, i.e. $\shape 1111$ can be formed by $2314$, but not by $2134$, and similarly (but the opposite) for $\shape 0220$.
As we can always reflect points in the plane over a dividing line, while discussing a shape we will not mention other shapes obtained by a sequence of such operations, because all such shapes can be counted in exactly the same way.
For instance $\shape3001,\shape0310,\shape1003$ and $\shape0130$ are all rotations of the same shape, but $\shape3010$ is not.
To sum up, there are the following possible shapes: $\shape4000,\shape3001,\shape3010,\shape2002,\shape2020,\shape1102,\shape1201,\shape1111$ and all their rotations.
We call shapes $\shape4000,\shape3100,\shape2020$ and their rotations \textit{non-proper}, because the division does not split the pattern both horizontally and vertically.
All other shapes are called \textit{proper}.
Now we are ready to state the crucial property that distinguishes two main groups of patterns:

\begin{definition}
A pattern that can form the shape $\shape 1111$ is called \textit{non-trivial}, and all other patterns are called \textit{trivial}.
\end{definition}

Notice that there are 8 trivial patterns: $1234,1243,2134,2143,4321,4312,3421,3412$, all other patterns are \textit{non-trivial}. 
All trivial patterns can form $\shape 0220$ (or its reflection $\shape 2002$), which cannot be formed by non-trivial patterns.
For a particular division of the plane, we say that an occurrence of a 4-pattern $\sigma$ is \textit{4-partite} if all its points belong to pairwise distinct regions, that is they form the shape $\shape 1111$.
To simplify notation, by counting 4-partite pattern $\sigma_4$ we mean counting 4-partite occurrences of the pattern $\sigma$.
Clearly, only non-trivial 4-patterns can be 4-partite.
We denote $\#_\sigma(P)$ as the number of occurrences of pattern $\sigma$ among the points from~$P$.
For a 4-partite pattern $\sigma_4$, we slightly abuse the notation and by $\#_{\sigma_4}\left(\shape{TL}{TR}{BL}{BR}\right)$ we denote the number of 4-partite occurrences of the pattern $\sigma_4$ in the plane divided into 4 regions: $TL,TR,BL,BR$.

\paragraph{MBRs.}
Let $\TT_n$ be a full binary tree with $n'=2^{\lceil\log n\rceil}$ leaves numbered from $1$ to $n'$ and with internal nodes corresponding to the range of indices of leaves from their subtrees.
We call the ranges corresponding to the nodes in the tree \textit{base ranges}.
Clearly, any number from $[n']$ is contained in $\log n'=\Oh(\log n)$ base ranges.
For a subset $S\subseteq [n]$, we define its \textit{minimum base range} $\MBR(S)$ as the smallest base range from $\TT_n$ containing all elements from~$S$.
Notice that it is the lowest common ancestor (LCA) of all leaves corresponding to the elements from~$S$.

We construct the full binary tree $\TT_n$ separately for $x$- and $y$-coordinates of points from $S_\pi$ and consider the Cartesian product $\TT_n\times \TT_n$ of the trees.
For every pair $(R_x,R_y)\in \TT_n \times \TT_n$ of ranges, let $P_{R_x,R_y}=\{(i,\pi(i))\in S_\pi: i\in R_x \wedge \pi(i)\in R_y\}$ be the set of points from $S_\pi$ with their coordinates in appropriate ranges.
We call a pair $(R_x,R_y)$ \textit{relevant} if its set $P_{R_x,R_y}$ is non-empty.
As every number belongs to $\Oh(\log n)$ base ranges, every point belongs to $\Oh(\log^2n)$ sets $P_{R_x,R_y}$ and hence we have:
\begin{observation}\label{obs:relevant_pairs}
There are $\Oh(n\log^2n)$ relevant pairs of ranges.
\end{observation}

\paragraph{General remarks.} All the reductions we show in this paper are split into several intermediate steps.
Unless stated otherwise, each presented reduction runs in time linear in the total size of the input and the sum of sizes of the created instances of the other problem we reduce to.

\subsection{Range Queries and Short Patterns}

Some of our algorithms use range queries for counting points in rectilinear (aligned with the $x$- and the $y$-axis) rectangles efficiently.
Below we provide the precise interface for such queries.

\begin{lemma}[\cite{Chazelle88,JaJaMS04}]\label{le:interface_range_queries}
 There exists a deterministic data structure that preprocesses a set of $n$ weighted points in $\Oh(n\log n)$ time and
 answers queries about the number or the sum of weights of points inside rectilinear rectangles in $\Oh(\log n)$ time.
\end{lemma}
\noindent
For completeness, we explain the folklore algorithms for counting patterns shorter than 4.

\begin{theorem}[cf.~{\cite[Corollary 2]{Even-ZoharL19}}]\label{thm:short_patterns}
 For any pattern $\sigma$, $|\sigma|<4$ there exists an algorithm counting $\sigma$ in permutations of length $n$ in $\Ohtilda(n)$ time.
\end{theorem}

\begin{proof}
Let $k=|\sigma|$.
Clearly, if $k=1$, we return $n$, the number of elements.
For $k=2$ and the pattern $12$ ($21$), for every element we count the number of larger (smaller) elements to the right, using a range query.
The precise interface for range queries used in this proof is provided in Lemma~\ref{le:interface_range_queries}.
Finally, if $k=3$ it suffices to show how to count patterns $123$ and $132$, because the other four patterns can be obtained from one of them after reversing and/or replacing every number $x$ with $4-x$.

\paragraph{123.}
We iterate through elements of $\pi$ and for each position $i$ we count occurrences of $123$ with the considered element $i$ as the middle one.
Let $x_i$ be the number of elements smaller than $\pi(i)$ to the left of $i$ and $y_i$ be the number of elements larger than $\pi(i)$ to the right of $i$, both these values can be obtained with a range query.
Then there are $x_i\cdot y_i$ occurrences of $123$ with the middle element at position $i$, so $\#_{123}(\pi)=\sum_{i=1}^nx_i\cdot y_i$.

\paragraph{132.}
We iterate through elements of $\pi$ and for each position $i$ we count pairs of elements to the right of $i$ which are larger than $\pi(i)$. 
This counts both the occurrences of $123$ and $132$, with the considered element $i$ as the first one.
Let $y_i$ be the number of elements larger than $\pi(i)$ to the right of $i$, which can be retrieved with a range query.
Then using the number of patterns $123$ computed in the previous paragraph we get:
$\#_{132}(\pi)=\sum_{i=1}^n\binom{y_i}{2} - \#_{123}(\pi)$.
\end{proof}

\subsection{Counting 4-Cycles}
Whenever we talk about counting 4-cycles in a graph we mean simple cycles (with all nodes distinct) of length 4, but not necessarily induced.
For counting 4-cycles self-loops and isolated nodes are irrelevant, but there might be multiple edges, and then we count the cycle (defined as a cyclic sequence of nodes) multiple times: the product of the multiplicities of the relevant edges.
Following the naming convention from \cite{LincolnWW18}, we define a \textit{4-circle-layered graph} to be a 4-partite directed graph with four disjoint groups of nodes $V_0,\ldots,V_3$ such that every edge in the graph is from the group $V_i$ to $V_{(i+1)\bmod 4}$ for some $0\leq i \leq 3$.

First, we show that, informally, counting 4-cycles in undirected graphs is equivalent to counting 4-cycles in 4-circle-layered graphs.
More precisely, we provide a sequence of reductions for counting 4-cycles in different graphs, starting from undirected graphs, through directed graphs to 4-circle-layered graphs and then back to undirected graphs.
We show that counting 4-cycles in a graph of each type can be reduced in $\Oh(m)$ time to a constant number of instances of counting 4-cycles in graphs of the next type.

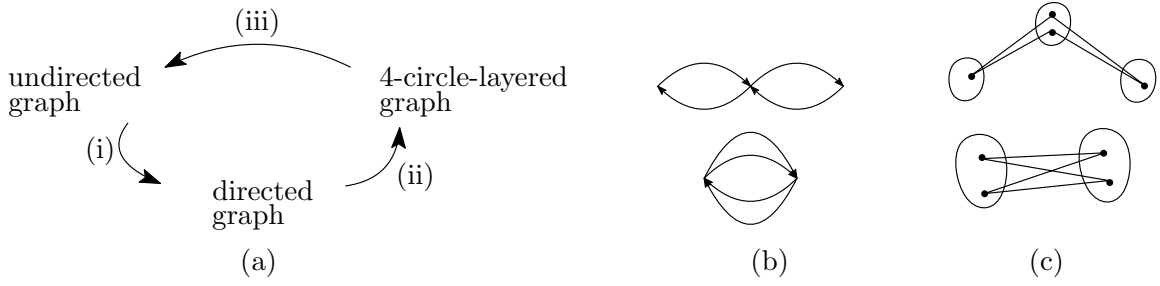
\FIGURE{h}{.95}{cycles}{(a) Sequence of reductions showing equivalence between counting 4-cycles in undirected, directed and 4-circle-layered graphs.
(b) Non-simple cycles from $G'$ to subtract in reduction (ii).
(c) Cycles to subtract (top) and add (bottom) in reduction (iii).}

\begin{lemma}\label{le:undirected-and-directed-4-partite}
Counting 4-cycles in undirected graphs on $m$ edges can be reduced to a constant number of instances of counting 4-cycles in 4-circle-layered graphs on $\Oh(m)$ edges and vice versa.
\end{lemma}
\begin{proof}
 We consider three types of graphs, first undirected graphs, then directed graphs and finally 4-circle-layered graphs.
 For each of them we show that counting 4-cycles in graphs of this type can be reduced in $\Oh(m)$ time to a constant number of instances of counting 4-cycles in the  graphs of the next type, as presented in Figure~\ref{fig:cycles}(a).
 We describe each of the reductions separately:
 \begin{enumerate}[(i)]
  \item \textbf{undirected $\rightarrow$ directed.} Given an undirected graph $G$ we construct a directed graph $G'$ replacing every undirected edge with two directed edges.
  Then the number of 4-cycles in $G'$ is twice the number of 4-cycles in $G$, as every cycle can be traversed in both directions.
  Then we have: $\#_{C_4}(G)=\frac12 \#_{C_4}(G')$.
  \item \textbf{directed $\rightarrow$ 4-circle-layered.} Given a directed graph $G'$ we construct a 4-circle-layered graph $G''$ by copying nodes of $G'$ four times and adding edges between corresponding nodes from two consecutive groups.
  More precisely, let $v_i''\in V_i''$ in $G''$ be the copy of node $v'$ from $G'$ in the $i$-th group.
  For every directed edge $(u',v')$ in $G'$ we add the edge $(u_i'',v_{(i+1)\bmod 4}'')$ to $G''$ for all $0\leq i \leq 3$.
  Then the number of 4-cycles in $G''$ is 4 times the number of 4-cycles in $G'$ plus some additional cycles which do not correspond to simple cycles in $G'$.
  More precisely, all the additional 4-cycles in $G''$ correspond to non-simple (on 2 or 3 distinct nodes) 4-cycles in $G'$, which are shown in Figure~\ref{fig:cycles}(b) and can be counted in linear time.
  Formally, let $b(u')=|\{v'\in V':(u',v')\in E' \wedge (v',u')\in E'\}|$ be the number of neighbors of a node $u'$ connected to $v'$ in both directions, which can be obtained by sorting the adjacency lists in linear time.
  Then we have: $\#_{C_4}(G')=\frac14\left(\#_{C_4}(G'')-\sum_{u'\in V'}\left(4\binom{b(u')}{2} + b(u')\right)\right)$.
  \item \textbf{4-circle-layered $\rightarrow$ undirected.} Given a 4-circle-layered graph $G''$ we create an undirected graph $G$ by undirecting all edges from $G''$.
  Then we can no longer ensure that the 4-cycles pass through 4 different groups of nodes, so we need to subtract 4-cycles fully contained in three groups of nodes and add 4-cycles fully contained in two groups, as shown in Figure~\ref{fig:cycles}(c).
  The number of such cycles can be obtained by counting 4-cycles in the graph $G$ restricted only to the particular groups of nodes.
  Formally, let $V_i$ be the group of nodes corresponding to $V_i''$ in $G''$ and $G[W]$ be the subgraph of $G$ restricted to the nodes from $W$ and edges between them.
  Then we have:
  $\#_{C_4}(G'')= \#_{C_4}(G)+ \sum_{0\leq i\leq 3}\#_{C_4}(G[V_i\cup V_{i+1}])-\#_{C_4}(G[V_i\cup V_{i+1}\cup V_{i+2}])$
  where the indices $i+1$ and $i+2$ are taken modulo 4.
  \qedhere
 \end{enumerate}
\end{proof}

A multigraph is a triple $(V,E,\mult)$, where $E$ is a set of $m$ edges and the function $\mult: E \rightarrow \{1,\ldots,U\}$ denotes multiplicity of an edge.
For simple graphs it holds that $\mult(e)=1$ for all edges $e\in E$ and the function is omitted.
Throughout this paper we focus mainly on simple graphs, but in one of the provided reductions we obtain a 4-circle-layered graph with multiplicities on every edge (or in other words, a 4-circle-layered multigraph), so in the following lemma we show how to reduce counting 4-cycles in such graphs to counting 4-cycles in 4-circle-layered simple graphs.

\begin{lemma}\label{le:dir_multi_to_simple}
  Counting 4-cycles in a 4-circle-layered multigraph with edge multiplicities bounded by~$U$ can be reduced to $\Oh(\log^4 U)$ instances of counting 4-cycles in 4-circle-layered simple graphs of the same size as the original graph. 
\end{lemma}
\begin{proof}
 Intuitively, we split every edge of the graph into edges with multiplicities being powers of two and iterate over all possible combinations of powers of two forming the cycle.
 
 More precisely, we iterate over all quadruples $(p_0,p_1,p_2,p_3) \in \{0,\ldots,\lfloor\log U\rfloor\}^4$ and for each of them create a simple, unweighted 4-circle-layered graph on the same set of nodes as the original graph and a subset of its edges.
 For all $0\leq i\leq 3$ we keep only the edges between groups $V_i$ and $V_{(i+1)\bmod 4}$ such that their multiplicity contains~$2^{p_i}$ in its binary representation.
 Then we count the number of 4-cycles in the obtained graph and multiply it by $2^{\sum_i p_i}$.
 Finally, the total number of 4-cycles in the original multigraph is the sum of results obtained for each quadruple.
\end{proof}
\section{Counting Patterns}

In this section we show that counting 4-partite patterns is equivalent, up to logarithmic factors, to counting 4-patterns.
The flavor of our arguments is similar to the ones used in \cite{Even-ZoharL19}, but we avoid the notion of corner tree formulas and explicitly state two technical lemmas that are required for our main result.
First we show that counting 4-partite patterns can be reduced to counting 4-patterns by omitting the division of the plane and using inclusion-exclusion principle.

\begin{lemma}\label{le:4-partite-pattern-2-pattern}
 Counting 4-partite pattern $\sigma_4$ on $n$ elements can be reduced to a constant number of instances of counting 4-pattern $\sigma$ in permutations of total size $\Oh(n)$.
\end{lemma}
\begin{proof}

When we omit the division of the plane and count the pattern $\sigma$ in the plane, we additionally count also the quadruples of points forming the pattern but coming from not all of the 4 regions of the plane.
To address this, we use inclusion-exclusion principle and add or subtract patterns on points from all possible subsets of regions.
Then the number of 4-partite patterns is:
$$\#_{\sigma_4}\left(\bigshape{TL}{TR}{BL}{BR}\right)=  \sum_{S\subseteq\{TL,TR,BL,BR\}} (-1)^{|S|}\cdot \#_\sigma\left(\bigcup_{Q\in S} Q\right)$$
where the union over regions chooses the specific subset of points preserving the relative order between them, as in the original setting.
\end{proof}

For the reduction in the other direction, first we need a technical lemma showing that all proper shapes but $\shape1111$ can be counted in $\Ohtilda(n)$ time.
Recall that we do not have to consider rotations of shapes separately, as they are equivalent under linear-time transformations of the input.

\newcommand{\REG}[4]{
\ifmmode
\frac{\mathsmaller{#1|#2}}{\mathsmaller{#3|#4}}
\else
$\frac{\mathsmaller{#1|#2}}{\mathsmaller{#3|#4}}$\xspace
\fi}
\newcommand{\BR}{\REG\square\square\square\blacksquare}
\newcommand{\BL}{\REG\square\square\blacksquare\square}
\newcommand{\TR}{\REG\square\blacksquare\square\square}
\newcommand{\TL}{\REG\blacksquare\square\square\square}

\begin{lemma}\label{le:easy_shapes}
 For any 4-pattern $\sigma$ and division of the plane with $n$ points, the shapes $\shape3001,\shape2002,\shape1102,\shape1201$ can be counted in $\Ohtilda(n)$ time.
\end{lemma}	
\begin{proof}
 To simplify the presentation, we use the graphical symbols to denote particular regions of the plane:\TL,\TR,\BL and\BR that denote $TL,TR,BL$ and $BR$ respectively.
 Notice the difference between the notion for 4-partite patterns $\sigma_4$ where $\#_{\sigma_4}\left(\shape{TL}{TR}{BL}{BR}\right)=\#_{\sigma_4}\left(\bigshape{\TL}{\TR}{\BL}{\BR}\right)$ and non-4-partite patterns $\sigma$, for which we use division of the plane only to specify the subset of points in which we count patterns, e.g. $\#_\sigma(\TL)=\#_\sigma(TL)$.
 In order to count shapes $\shape3001$ and $\shape2002$ it suffices to count appropriate 3-, 2- or 1-patterns on points in\TL or\BR and multiply the two numbers.
 By Theorem~\ref{thm:short_patterns}, this approach runs in $\Ohtilda(n)$ time.
 
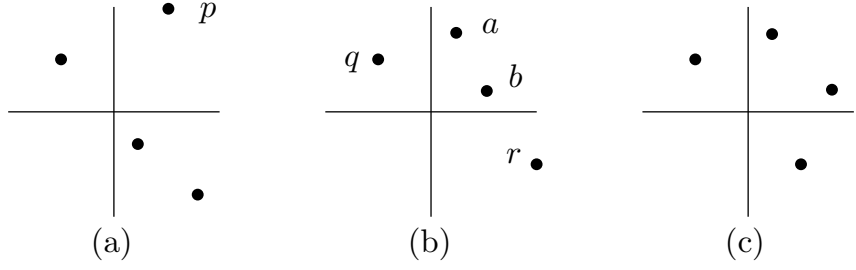
\FIGURE{h}{.7}{simple-shapes}{(a) A quadruple of points forming $\shape1102$, where the two bottom points alone form the pattern~$21$. (b) Naming of points in $\shape1201$. (c) $3412$ is the most difficult pattern to count.}
 
 Now we show how to count the shape $\shape1102$.
 Suppose that in the pattern $\sigma$, the two points in\BR form the pattern $21$, see Figure~\ref{fig:simple-shapes}(a) for an example.
 For the other case of the pattern $12$ we can apply horizontal reflection for points in both the bottom regions.
 First we preprocess\BR and for every point there we count points from\BR ``to the right and down'' of it and ``to the left and up'' using range queries.
 The precise interface for range queries used in this proof is provided in Lemma~\ref{le:interface_range_queries}.\
 Next, we iterate over all points $p$ in\TR and for each of them need to count points in\TL and pairs of pairs of points in\BR that together form the pattern $\sigma$.
 The former number is computed with a range query about the number of points from\TL that are below or above $p$, depending on $\sigma$.
 To compute the latter number, notice that the point $p$ can be in three positions with respect to the two points from\BR: either to the left of both of them, to the right or in-between (as in Figure~\ref{fig:simple-shapes}(a)).
 Each of the cases can be retrieved by either:
 \begin{enumerate}[(a)]
  \item \label{it:rd} counting points ``to the right and down'' for all points from\BR to the right of $p$, or
  \item \label{it:lu} counting points ``to the left and up'' for all points from\BR to the left of $p$, or
  \item subtracting the values obtained in (\ref{it:rd}) and (\ref{it:lu}) from $\#_{21}(\BR)$, the total number of pairs of points from\BR, such that one of them is ``to the right and down'' from the other.
 \end{enumerate}
 All the above values can be obtained in $\Oh(\log n)=\Ohtilda(1)$ time with range queries about the sum of weights of points in a rectangle.
 
 Counting the shape $\shape1201$ is slightly more involved as now we do not have a single ``central'' region in which we can iterate over points and obtain the answer, as it was the case with points $p\in\TR$ for the shape $\shape1102$.
 In order to refer to the points more easily, we use the naming of points as in Figure~\ref{fig:simple-shapes}(b), that is $q$ is the point from\TL, $r$ from\BR and $a$ and $b$ from \TR, where $a$ is to the left of $b$.
 Again we focus only on the case when points from\TR form the pattern $21$, that is $a$ is ``to the left and up'' of $b$.
 For the other case of the pattern $12$ we can horizontally reflect points in both the top regions. 
 Consider the case when the last element in the pattern $\sigma$ is the smallest (equals 1), so is the point $r$, in\BR.
Then the allowed location of $r$ depends only on the point $b$, as $r$ must be to the right of $b$, so for every point $b$ in\TR we can count points from\BR that are to the right of~$b$.
Next we proceed similarly as while counting the shape $\shape1102$, that is we iterate through points $q$ from\TL and count pairs of points $a$ and $b$ in the appropriate order with respect to $q$, where additionally points $b$ have weights.
 
The above approach can be also applied to all shapes in which~$r$ is to the left of both points $a$ and $b$, or $q$ is below both $a$ and $b$,
or $q$ is above both $a$ and $b$.
In other words, this covers all patterns in which $q$ is not between $a$ and $b$ or $r$ is not between $a$ and $b$.
 Hence it remains to consider the patterns in which both $q$ and $r$ are between $a$ and $b$.
 Notice that for the fixed relation between points $a$ and $b$ ($21$ in our case) there is exactly one such pattern $\sigma$: $3412$, see Figure~\ref{fig:simple-shapes}(c).
 To sum up, there are 9 possible patterns (3 locations for points $r$ and $q$ are possible independently) forming the considered shape $\shape1201$ and 8 of them we can count in $\Ohtilda(n)$ time.
 Moreover, the sum of counts of all the 9 patterns is exactly $|\TL|\cdot|\BR|\cdot\#_{21}(\TR)$.
 Subtracting from the total count the 8 values that we can compute efficiently gives us the number of occurrences of the last pattern.
 Thus, all patterns forming the shape $\shape 1201$ can be counted in $\Ohtilda(n)$ time.
\end{proof}

Recall that, given a division of the plane into 4 regions, an occurrence of a 4-pattern $\sigma$ is 4-partite if all its elements are in pairwise distinct regions.
In the following lemma we show that we can count 4-patterns by counting 4-partite patterns.
At a high level, every occurrence of the pattern is counted while considering the division of the plane aligned with the division of minimum base ranges containing all  coordinates of the four points.

\begin{lemma}\label{le:pattern-2-4-partite-pattern}
 Counting a 4-pattern $\sigma$ on $n$ elements can be reduced in $\Ohtilda(n)$ time to multiple instances of counting 4-partite patterns $\sigma_4$ of total size $\Ohtilda(n)$.
\end{lemma}
\begin{proof}
Recall that $\MBR(S)$, the minimum base range  of a set $S\subseteq [n]$ is the minimum base range containing all elements of $S$ in the full binary tree $\TT_n$ on $n'=2^{\lceil \log n \rceil}$ leaves and $R\in \TT_n$ is a set of consecutive elements from $[n]$.
By Observation~\ref{obs:relevant_pairs} we have that there are $\Ohtilda(n)$ pairs $(R_x,R_y)\in \TT_n\times\TT_n$ for which there exists an $i\in[n]$ such that $i\in R_x$ and $\pi(i)\in R_y$.
We can retrieve all such pairs in $\Ohtilda(n)$ time by iterating through all points from $S_\pi$ and generating the set of all relevant pairs of ranges.
Recall that $P_{R_x,R_y}=\{(i,\pi(i))\in S_\pi: i\in R_x \wedge \pi(i)\in R_y\}$.
In terms of the permutation $\pi$, $R_x$ corresponds to its substring and $R_y$ restricts its values.

For every relevant pair of ranges $(R_x,R_y)$ with $P_{R_x,R_y}$ of at least 4 points inserted, we consider the plane restricted only to points from $P_{R_x,R_y}$ and divided in the following way.
As all points from $S_\pi$ have distinct coordinates and $|P_{R_x,R_y}|\geq 4$, the range $R_x$ contains at least 4 elements, so is not a leaf in $\TT_x$ and has two children $R_x^L,R_x^R$ in $\TT_x$.
The two ranges $R_x^L$ and $R_x^R$ are disjoint so we can find a vertical line that separates them, i.e. that passes through the middle of segment between the rightmost element from $R_x^L$ and the leftmost element from $R_x^R$.
Notice that this line does not pass through a point from $P_{R_x,R_y}$ as $R_x^L$ and $R_x^R$ are two consecutive ranges in $\TT_n$.
We find a horizontal line separating the range $R_y$ in the same way.
For the set of points $P_{R_x,R_y}$ and the above division of the plane, we count all shapes $\shape3001,\shape2002,\shape1102,\shape1201$ and all their possible rotations in $\Ohtilda(|P_{R_x,R_y}|)$ time, by Lemma~\ref{le:easy_shapes}.
Finally, we need to count the shape $\shape1111$, the 4-partite pattern $\sigma_4$ on the set $P_{R_x,R_y}$ and sum up all the obtained results.

Now we show that the above procedure counts every occurrence of the pattern $\sigma$ exactly once, while considering the pair of minimum base ranges for both coordinates of the points from the occurrence.
Formally, an occurrence $g$ of $\sigma$ on positions $i_1<i_2<i_3<i_4$ is counted only for the pair of ranges $(R_x,R_y)$ where $R_x=\MBR(\{i_1,i_2,i_3,i_4\})$ and $R_y=\MBR(\{\pi(i_1),\pi(i_2),\pi(i_3),\pi(i_4)\})$ and the appropriate shape, depending on the position of points from $\{(i_j,\pi(i_j)):j\in[4]\}$ with respect to the division.
Suppose the contrary, that $g$ is counted for another pair of ranges $(R_x',R_y')$ where $R_x'\ne R_x$, for $R_y'\ne R_y$ the reasoning is similar.
If $\{i_1,i_2,i_3,i_4\}\not\subseteq R_x'$, then for some $j$ the point $(i_j,\pi(i_j))$ will not be present in the considered instance.
Otherwise, from the structure of base ranges we have that $\MBR(\{i_1,i_2,i_3,i_4\})$ is fully contained in one half of $R_x'$.
In this case $g$ also will not be counted, because it forms a non-proper shape for the considered division ($\shape 2020$, $\shape 3100$ or $\shape 4000$ or their rotations) and we do not count such shapes.

As every point from $S_\pi$ is included in $\Oh(\log^2n)$ sets $P_{R_x,R_y}$, the total size of all the considered sets is $\Ohtilda(n)$ and hence counting shapes different than $\shape 1111$ takes $\Ohtilda(n)$ time.
Similarly, the total size of the instances of counting 4-partite pattern $\sigma_4$ is $\Ohtilda(n)$.
\end{proof}

By definition, trivial patterns do not form the $\shape1111$ shape, so the reduced instances have always 0 occurrences of the 4-partite pattern, which can be returned in constant time. Hence:

\begin{corollary}[cf.~{\cite[Corollary 3]{Even-ZoharL19}}]
 All trivial 4-patterns ($1234, 1243, 2134, 2143, 4321, 4312,$ $3421, 3412$) in permutations of length $n$ can be counted in $\Ohtilda(n)$ time.
\end{corollary}

\section{Equivalence of Counting 4-Partite Patterns and Cycles}

First we show that in fact all (non-trivial) 4-partite patterns are equivalent by a linear-time transformation of the considered set of points.
At a high level, we will show that reversing the order of points in any of the four parts of the plane (left, top, ...) allows us to slightly modify the pattern.

\begin{lemma}\label{le:all_patterns_equiv}
 Counting any non-trivial 4-partite pattern $\sigma_4$ can be reduced to counting any other non-trivial 4-partite pattern $\sigma'_4$.
\end{lemma}
\begin{proof}
We start with showing that by reversing the points in the left part of the plane we can swap the first two elements of the pattern:

$$\#_{abcd_4}\left(\bigshape{TL}{TR}{BL}{BR}\right)=\#_{bacd_4}\left(\rev\hspace{-3px}\left(\frac{TL}{BL}\right)\hspace{-5px}\frac{|TR}{|BR}\right).$$
Formally, suppose that we need to count the 4-partite pattern $abcd$ in the plane divided as follows: $\shape{TL}{TR}{BL}{BR}$ and the leftmost and rightmost points from the left part ($TL\cup BL$) have the $x$-coordinate respectively $x_1$ and $x_2$.
We replace every point $(x,y)$ from the left part with $(x_1+x_2-x,y)$.
Then, only the horizontal order of points from the left part is reversed and any 4-partite occurrence of the pattern $abcd$ in the original instance corresponds to a 4-partite occurrence of the pattern $bacd$ in the transformed instance.
Similarly, after reversing the right part we obtain the pattern $abdc$ from $abcd$.
When we reverse the (vertical) order of the top or bottom part, we swap respectively elements $3$ and $4$ or $1$ and $2$ in the pattern.
For example, by reversing the top part, from the pattern $1324$ we obtain the pattern $1423$.

Observe that operations in any two parts of the plane are independent, we can apply any subset of them and obtain either of the 16 possible non-trivial 4-partite patterns.
See Figure~\ref{fig:all_patterns_equiv} %
with the precise description of operations between the patterns.
Thus, we can transform  in linear time any instance of counting non-trivial 4-partite pattern $\sigma_4$ to an instance of counting either of the 16 possible non-trivial 4-partite patterns.
\end{proof}

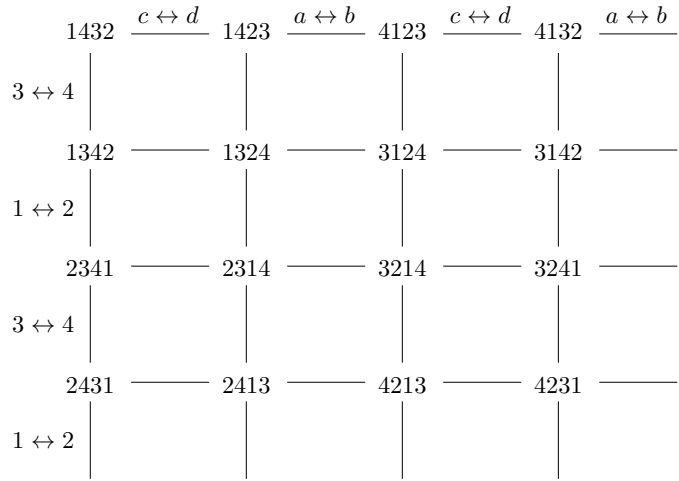
\FIGURE{h}{.55}{all_patterns_equiv}{Reductions between non-trivial patterns described in Lemma~\ref{le:all_patterns_equiv}.
Operation $a\leftrightarrow b$ ($c\leftrightarrow d$) swaps the first (second) pair of elements in the pattern and corresponds to reversing left (right) part of the plane.
Operation $1\leftrightarrow 2$ ($3\leftrightarrow 4$) swaps elements 1 and 2 (3 and 4) in the pattern and corresponds to reversing bottom (top) part of the plane.}

Hence in the following claims it suffices to consider only one non-trivial 4-partite pattern and we will focus on counting the pattern $1324_4$.
Notice that in $\Ohtilda(n)$ time we can shift any set of $n$ points in such a way that the division lines are aligned with $x$- and $y$- axes and all points have integer coordinates from $\NN=\{-n,\ldots,-1,1,\ldots,n\}$, preserving the relative order between the parts.
In the following lemma we show that counting non-trivial 4-partite patterns can be reduced to counting 4-cycles in 4-circle-layered multigraphs.
At a high level, we will group all occurrences of the pattern by the minimum base ranges of coordinates of points in each of the parts of the plane.
  
\begin{lemma}\label{le:pattern-2-cycle}
Counting a non-trivial 4-partite pattern on $n$ points can be reduced to an instance of counting 4-cycles in a 4-circle-layered multigraph on $\Ohtilda(n)$ edges with multiplicities bounded by $n$. 
\end{lemma}
\begin{proof}
For a permutation $\pi$ and division of the plane with points $S_\pi$ we need to construct a 4-circle-layered multigraph in such a way that the number of 4-cycles in the graph gives us the number of occurrences of the pattern.
Recall that we can operate on points from $\NN^2$ and the division of the plane along the $x$- and $y$-axes.
We consider four full binary trees $\TT_n^L,\TT_n^R,\TT_n^B,\TT_n^T$ for each part of the plane separately.
For each base range in the trees we create a separate node in the new 4-partite graph.

\FIGURE{h}{.9}{patterns-2-cycles}{
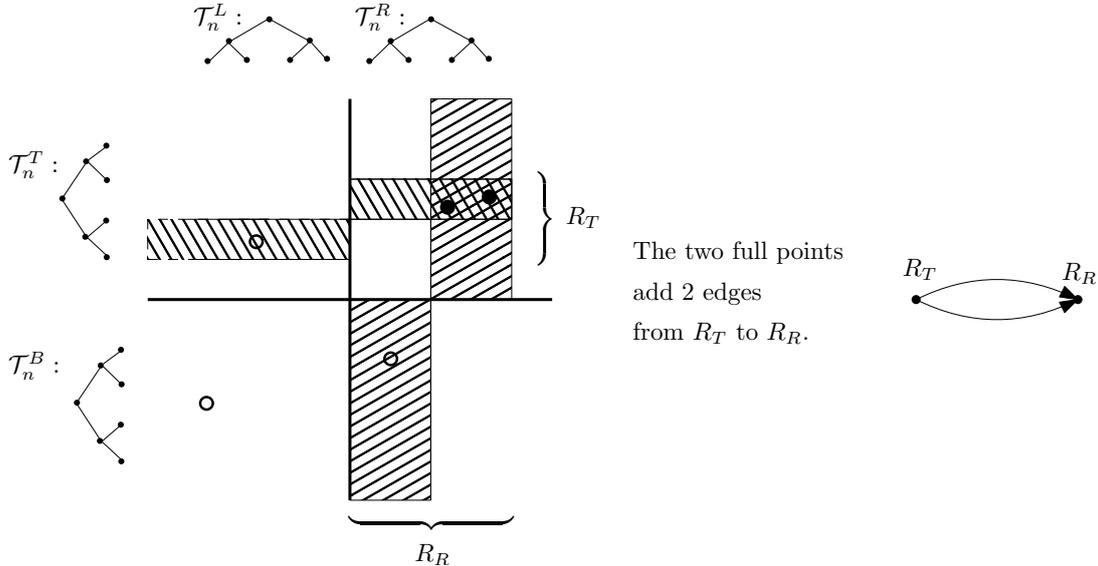
We consider four full binary trees $\TT_n^L,\TT_n^R,\TT_n^B,\TT_n^T$ for each part of the plane separately and group occurrences of patterns by the $\MBR$s of coordinates in each part of the plane.
Points from appropriate halves of MBRs from each two consecutive parts add a new edge to the multigraph.
}

Now we process all points from $S_\pi$ grouped by their region.
Suppose we process a point $(x,y)\in S_\pi$ from the top-right region.
We iterate over all pairs $(R_R,R_T)\in \TT_n^R\times\TT_n^T$ of base ranges such that $x\in R_R$ and $y\in R_T$ and the ranges are not singletons (leaves in $\TT_n$), so contain at least two elements from $[n]$.
Recall that we focus on the pattern $1324$, because now the choice of the particular pattern is crucial in the following condition.
We add edge $(R_T,R_R)$ to the 4-circle-layered multigraph if $x$ is in the right half of $R_R$ and $y$ is in the top half of $R_T$.
This means that the point $(x,y)$ can be a part of an occurrence of the $1324$ pattern in which $R_T$ is the $\MBR$ of $y$-coordinates of the top points and $R_R$ is the $\MBR$ of $x$-coordinates of the right points.
See Figure~\ref{fig:patterns-2-cycles}.
We proceed similarly for the remaining three regions, modifying only the condition for including an edge, based on the position of elements of the pattern $1324$ inside the considered region.

If an edge is inserted more than once, we simply increment its multiplicity, which can be stored e.g. in a balanced binary search tree.
As every point from $S_\pi$ adds at most $\Oh(\log^2n)$ edges, in total there are $\Oh(n\log^2n)=\Ohtilda(n)$ edges in the graph.
Clearly, the constructed directed multigraph is 4-partite as we connect nodes from $\TT_n^T$ to the nodes from $\TT_n^R$, from $\TT_n^R$ to $\TT_n^B$ etc.
Finally, observe that the multiplicity of an edge connecting nodes corresponding to ranges $R$ and $R'$ is the number of points in the intersection of their appropriate halves.
Hence multiplicities of edges in the graph are bounded by~$n$.
\end{proof}

The reduction from 4-circle-layered multigraphs to 4-circle-layered simple graphs was shown in Lemma~\ref{le:dir_multi_to_simple}.
Finally, to conclude the equivalence between counting 4-partite patterns and cycles in 4-circle-layered graphs, we describe the reduction from counting 4-cycles in 4-circle-layered graphs to counting non-trivial patterns.
The idea is to first embed the graph in the plane so that every group $V_i$ of nodes corresponds to a half-plane and edges to points in the plane.
Then every 4-cycle corresponds to a rectangle with all corners in distinct quadrants.
Now we appropriately tilt each quadrant, so that every rectangle corresponds to an occurrence of the pattern $1324_4$.
However, this change introduces many more occurrences of the pattern as now we have slightly weaker constraints on the relative position of points.
This is corrected by subtracting the surplus by applying the inclusion-exclusion principle for different ways of tilting the quadrants.%

We remark that our approach is similar to that of Berendsohn et al.~\cite[Section 5]{BerendsohnKM19}.
They showed a reduction from Partitioned Subgraph Isomorphism to counting short patterns in permutations by embedding the input graph in the plane with appropriate tilting and using the inclusion-exclusion principle.
However, while their reduction works for arbitrary subgraphs of size~$k$, this comes at the cost of increasing the size of the
permutation pattern to $7k+1$, which in our case would result in a permutation pattern on 29 elements, hence
not giving us the desired tight connection between counting 4-cycles and 4-patterns.

\begin{lemma}\label{le:cycle-2-pattern}
 Counting 4-cycles in a 4-circle-layered simple graph on $m$ edges can be reduced in $\Ohtilda(m)$ time to a constant number of instances of counting a non-trivial pattern in a permutation of length~$m$.
\end{lemma}
\begin{proof}
Given a 4-circle-layered graph $G=(V_0\dot\cup V_1\dot\cup V_2\dot\cup V_3,E)$, where $E\subseteq \bigcup_i V_i\times V_{i+1\bmod 4}$, we will embed it in the plane and construct a constant number of instances of counting a non-trivial 4-partite pattern.
As Lemma~\ref{le:all_patterns_equiv} guarantees that all such patterns are equivalent, we can focus only on the pattern $1324$.

Every half-plane corresponds to a part of the graph in the clockwise order:
negative $x$-coordinates correspond to nodes from $V_0$, positive $y$-coordinates correspond to nodes from $V_1$, positive $x$-coordinates correspond to nodes from $V_2$ and negative $y$-coordinates correspond to nodes from $V_3$.
The order of points in every half-plane projected on the appropriate axis is arbitrary, so we can use any injective mapping from $V_0$ and $V_3$ to $\{-n,\ldots,-1\}$ and from $V_1$ and $V_2$ to $\{1,\ldots,n\}$.
Next, every edge in the graph corresponds to a point in the plane, so we get a subset of $m$ points from $\NN^2$.
Then every 4-cycle in $G$ corresponds to a rectangle with corners in points in distinct quadrants.

Now we would like to transform the constructed set of points into a number of point sets $S_\pi$ for some permutations $\pi$.
Intuitively, every 4-cycle from $G$ will correspond to an occurrence of the pattern $1324_4$.
Notice that there might be many edges incident to a node, so in the beginning some points have equal $x$- or $y$-coordinate, which we need to avoid.
At first we will guarantee that no two points from distinct quadrants have equal $x$- or $y$-coordinates, which is already sufficient to be able to define an occurrence of the 4-partite pattern $1324_4$.
In the end we will show that we can slightly shift all points preserving relationships between points from distinct quadrants and additionally ensuring uniqueness of coordinates inside each quadrant.
Consider the following transformation of the plane:

$$\bigshape{TL}{TR}{BL}{BR} \rightarrow \bigshape{TL+(\frac15,0)}{TR+(0,\frac15)}{BL+(0,-\frac15)}{BR+(-\frac15,0)}$$
where by adding a vector to the region we denote shifting all points from the region by the vector.
Informally, we shift $TL$ slightly right, $TR$ slightly up etc, see Figure~\ref{fig:tilting}(a).
Observe that now every 4-cycle from $G$ corresponds to an occurrence of $1324_4$ (see Figure~\ref{fig:tilting}(b) and its explanation in the caption), but there are also many more other occurrences of the pattern, which do not correspond to a cycle from~$G$.
More precisely, every occurrence of the pattern $1324_4$ corresponds to 4 edges from~$G$, but we cannot ensure that they form a cycle, or equivalently, that every two consecutive edges share an endpoint, see Figure~\ref{fig:tilting}(c).

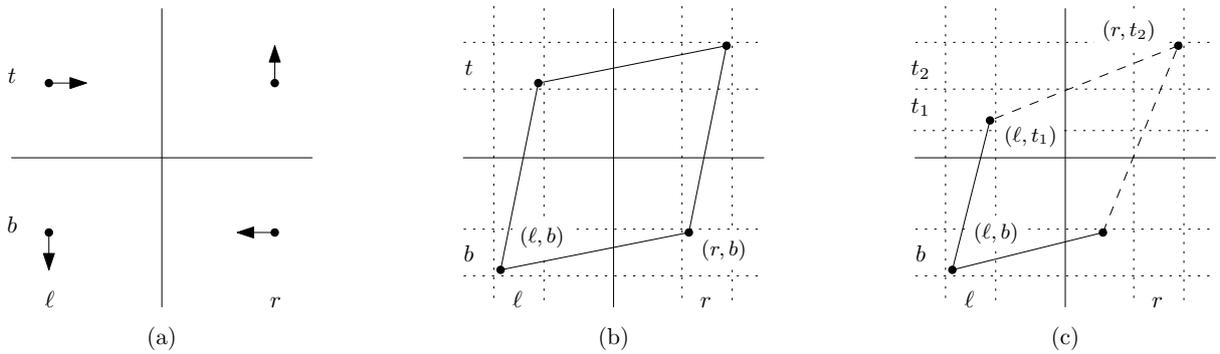
\FIGURE{h}{1}{tilting}{(a) Slightly shifting all points guarantees that points from distinct quadrants do not share a coordinate.
(b) Every cycle from the graph corresponds to an occurrence of $1324_4$.
We mark the area of the ``small shifts'' between the dashed lines, so i.e. all points that initially had $y$-coordinate equal to $t$ now are between the two horizonatal dashed lines surrounding $t$.
(c) Some occurrences of $1324_4$ do not correspond to a cycle in $G$, as the consecutive edges do not share endpoints.
Points corresponding to consecutive edges that share an endpoint are connected with a solid line (i.e. $(\ell,b)$ and $(\ell,t_1)$) and with a dashed line if they do not share (i.e. $(\ell,t_1)$ and $(r,t_2)$).}

In particular, after the above transformation, in every occurrence of $1324_4$ the two points from the left half-plane: $(x,y-\frac15)\in BL$  and $(x'+\frac15,y')\in TL$ satisfy that $x'+\frac15\geq x$, but the two edges corresponding to these points share an endpoint only when $x'= x$.
On the other hand, if we slightly modify the above transformation and set $TL \rightarrow TL+(-\frac 15,0)$, we obtain that $x'-\frac15\geq x$, so $x'>x$ and certainly the two edges cannot share an endpoint.
Now we use this property for all half-planes and plug the modified transformations into the inclusion-exclusion principle:
$$\#_{C_4}(G)=\sum_{S\subseteq\{L,R,T,B\}} (-1)^{|S|}\#_{1324_4}\left(\bigshape
{\transform{TL}{\delta_L(S)}{0}}
{\transform{TR}{0}{\delta_T(S)}}
{\transform{BL}{0}{-\delta_B(S)}}
{\transform{BR}{-\delta_R(S)}{0}} \right) $$
where $\delta_X(S)=\frac15\text{ if }X\in S$ or $-\frac15$ otherwise.
Finally, to ensure that no two points in a single quadrant have equal $x$- or $y$-coordinate we first transform every point $(x,y)$ into $(x+\frac{y}{10n},y+\frac{x}{10n})$ and then shift accordingly.
For instance, a point $(x,y)\in TL$ is transformed to $(x+\frac{y}{10n}+\delta_L(S),y+\frac{x}{10n})$.
Notice that the choice of lengths of the shifts guarantees that no two points have the same $x$- or $y$-
coordinate and the new coordinates are within $[-\frac{3}{10},\frac{3}{10}]\times [-\frac{3}{10},\frac{3}{10}]$ square comparing to the original location of points.
In the obtained instances all points have non-integer coordinates, but we can normalize them into $\NN^2$ preserving the relative order between the points.
\end{proof}

\printbibliography
\end{document}